\documentclass[a4paper,10pt]{article}
\usepackage[utf8]{inputenc}

\usepackage{graphicx}
\usepackage{amsthm,amsfonts,amsmath} 

\newcommand{\e}{\epsilon}

\newcommand{\be}{\begin{equation}}
\newcommand{\ee}{\end{equation}}

\newcommand{\ba}{\begin{eqnarray}}
\newcommand{\ea}{\end{eqnarray}}
\newtheorem{prop}{Proposition}

\title{Brief Article}

\begin{document}
\centerline{\textbf{All transverse and TT tensors in  flat spaces of any dimension}}

\date{\today}

\null

\centerline{ J Tafel} 
\noindent
{Institute of Theoretical Physics, University of Warsaw, Poland}

\begin{abstract}
We present general formulas for transverse and transverse-traceless (TT) symmetric  tensors in flat spaces. TT tensors in conformally flat spaces can be obtained by means of a conformal transformation.

\end{abstract}

\null


\null


\section{Introduction}
A symmetric tensor $T^{ij}$  is called transverse if it satisfies 
\be
\nabla_iT^{ij}=0\label{1}
\ee
where $\nabla_i$ denotes the covariant derivative with respect to metric $g_{ij}$ of a $D$-dimensional space. Condition (\ref{1})
occurs in general relativity as an analog of the conservation law for energy and momentum, as a gauge condition for the metric tensor or as the momentum constraint in the initial data problem for the vacuum Einstein equations. In the last case one often assumes that, additionally,
\be\label{2}
g_{ij}T^{ij}=0\ .
\ee
 If (\ref{1}) and  (\ref{2}) are satisfied then the tensor $T$ is called  transverse-traceless (TT). TT tensors were introduced by Arnowitt, Deser and Misner in 1959 \cite{adm}.

Equation (\ref{1}) is particularly simple if $\nabla_i$ is the covariant derivative corresponding to a flat connection. Even then its solutions are interesting, especially if   (\ref{2}) is satisfied, since they can be used to construct initial data with a conformally flat metric by means of  the Lichnerowicz-York method \cite{l,y}. This method is based on the invariance of equations (\ref{1})-(\ref{2}) under the transformation
\begin{equation}\label{1a}
g'_{ij}=\psi^{\frac{4}{D-2}} g_{ij},\quad T'_{ij}=\psi^{-2} T_{ij}
\end{equation}
with a nonvanishing function $\psi$. Examples of data obtained in this way can be found in classical papers by Misner \cite{m}, Brill and Lindquist \cite{bl}, Bowen and York \cite{by} and Brandt and Br\"{u}gmann \cite{bb}.

In this article we present general solution of equations
\be
T^{ij}_{\ ,j}=0\ ,\label{3}
\ee
 with or without condition  (\ref{2}) with constant metric $g_{ij}$. They are expressed in terms of potentials which undergo gauge transformations. From the  traceless solutions   we can obtain all TT tensors in conformally flat spaces by means of transformation (\ref{1a}).

 For obvious reasons we assume $D\geq 2$. All results are based on the Poincare lemma about local exactness of closed forms. For our purposes it is convenient to use the following form of this lemma
 \be\label{0}
 \alpha^{Ii_1..i_p}_{\ \ \ \ \ ,i_p}=0\ \ \Rightarrow\ \ \alpha^{Ii_1..i_p}=\beta^{Ii_1..i_{p+1}}_{\ \ \ \ \ \ ,i_{p+1}}\ ,
 \ee
where tensors $\alpha$ and $\beta$ are completely antisymmetric in indices $i_1,i_2,...$ and  $I$ denotes a set of additional indices. Integer $p$ satisfies $1\leq p\leq D-1$.
\section{Solutions of $T^{ij}_{\ ,j}=0$}
\begin{prop}\label{p1}
 All solutions of equation (\ref{3}) have the form
 \be\label{4}
 T^{ij}=R^{ikjp}_{\ \ \ \ ,kp}\ ,
 \ee
 where $R^{ikjp}$ is a tensor  possessing all algebraic symmetries of the Riemann tensor.
  \end{prop}
\begin{proof}
Let us apply the Poincare lemma (\ref{0}) with $p=1$ and $I=i$ to equation (\ref{3}). Hence, there exists tensor $S^{ijk}$ such that 
\be\label{5}
T^{ij}=S^{ijk}_{\ \ \ ,k}\ ,\ \ S^{ikj}=-S^{ijk}\ .
\ee
Condition $T^{ij}=T^{ji}$ yields
\be
S^{[ij]k}_{\ \ \ \   ,k}=0\ .
\ee
Taking into account  (\ref{0}) with $I=[ij]$ we obtain potentials $V^{ijkp}$ such that 
\be
S^{[ij]k}=V^{ijkp}_{\ \ \ \ ,p}\label{6}
\ee
and
\be
V^{ijkp}=-V^{jikp}=-V^{ijpk}\ .
\ee
A combination of equations (\ref{6}) for different permutations of indices $ijk$ leads to
\be\label{8}
S^{ijk}=(V^{ijkp}-V^{ikjp}-V^{jkip})_{,p}\ .
\ee
Substituting (\ref{8}) to (\ref{5}) yields
\be\label{9}
T^{ij}=-(V^{ikjp}+V^{jpik})_{,kp}\ .
\ee
Let $R^{ikjp}$ be defined by 
$$R^{ikjp}=-V^{ikjp}-V^{jpik}+2V^{[ikjp]}\,$$
where the square bracket denotes antisymmetrization over all indices inside the bracket. Then (\ref{9}) takes the form
 (\ref{4}). Tensor $R^{ijkp}$ has algebraic symmetries of the Riemann tensor. It is antisymmetric in the first and  second pair of indices, symmetric with respect to interchange of these pairs and $R^{[ijkp]}=0$ (hence also $R^{i[jkp]}=0$).

\end{proof}
For $D=2$ one has
\be\label{9a}
R^{ikjp}=\frac 12R\epsilon^{ik}\epsilon^{jp}\ ,
\ee
where $R$ is arbitrary function and $\e^{ij}$ is the standard completely antisymmetric tensor. 
It follows from (\ref{4}) that
\be\label{9b}
 T^{ij}=\epsilon^{ik}\epsilon^{jp}R_{,kp}\ .
 \ee
 For $D\geq 3$ tensor $R^{ijkp}$ contains more components than number of degrees of freedom admitted by equation (\ref{3}).  Arbitrariness in choice of $R^{ijkp}$ describes the following proposition.
\begin{prop}\label{p2}
For $D=2$ function  $R$ is given up to the tranlation\\
$R\longrightarrow R+R'$, where
\be\label{9a}
R'=a_ix^i+b\ ,\ \ \ a_i,b=const\ .
\ee
For $D\geq 3$ tensor $R^{ikjp}$ is given up to the translation by
 \be\label{10}
 R'^{ikjp}=(\frac 12\xi^{ikjpr}+\frac 12\xi^{jpikr}-\xi^{[ikjpr]})_{,r}\ ,\ \ \ \xi^{ikjpr}=\xi^{[ik][jpr]}\ .
 \ee
 \end{prop}
\begin{proof}
 For  $D=2$ an addition to $R$ does not change  $T^{ij}$ if all its second derivatives vanish. Thus, it has to be linear in all coordinates. 

In order to show that for  $D\geq 3$ all solutions of the equation
\be\label{11}
 R'^{ikjp}_{\ \ \ \ ,kp}=0\ ,
 \ee
 with symmetries of the Riemann tensor, have the form  (\ref{10}) let us write (\ref{11}) as
 \be\label{12}
 (R'^{ikjp}_{\ \ \ \ ,k})_{,p}=0\ .
 \ee
 From the Poincare lemma (\ref{0}) with $I=i$ one obtains
 \be\label{12}
 R'^{ikjp}_{\ \ \ \ ,k}=\frac 12V^{ijpk}_{\ \ \ \ ,k}\ ,
 \ee
 where tensor $V^{ijpk}$ is antisymmetric in the last 3 indices, $V^{ijpk}=V^{i[jpk]}$. 
Thanks to  $R'^{[ijk]p}=0$ an antisymmetrization of (\ref{12}) over the indices $ij$ leads  to
\be\label{13}
 (R'^{ijpk}+V^{[ij]pk})_{,k}=0\ .
 \ee
 Using again (\ref{0}), now with $I=[ij]$, we obtain
 \be\label{14}
 R'^{ijpk}+V^{[ij]pk}=\xi^{ijpkr}_{\ \ \ \ \ ,r}\ ,
 \ee
 where $\xi^{ipjkr}=\xi^{[ip][jkr]}$.
 Antisymmetrization of (\ref{14}) over $ijp$ yields
 \be\label{15}
 V^{[ijp]k}=\xi^{[ijp]kr}_{\ \ \ \ \ \ \ ,r}\ .
 \ee
 Taking a suitable combination of equations (\ref{15}) with different permutations of indices allows to express   tensor $V$ in the form
 \be\label{16}
 V^{ijpk}=(3\xi^{[jpk]ir}+4\xi^{[ijpk]r})_{,r}\ .
 \ee
 Substituting (\ref{16}) to (\ref{14}) leads to (\ref{10}).
 
 \end{proof}
  In dimension $D=2$ gauge transformations (\ref{9a})  can be used to remove a constant term and terms linear in coordinates in an expansion of $R$ around a fixed point. In higher dimensions tensor $R^{ikjp}$ can be decomposed into a traceless part  $C^{ikjp}$, corresponding to the Weyl tensor in general relativity, and the rest which is defined  by an analog of the Ricci tensor $R_{kp}=R^i_{\ kip}$
  \be\label{16a}
  R^{ik}_{\ \ jp}=C^{ik}_{\ \ jp}+4a \delta^{[i}_{\ [j}R^{k]}_{\ p]}-2b R\delta^{[i}_{\ [j}\delta^{k]}_{\ p]}\ .
\ee  
Here $R=R^i_{\ i}$ and
\be\label{16b}
a=\frac{1}{D-2}\ ,\ \ b=\frac{1}{(D-1)(D-2)}\ .
\ee
Tensor $C^{ikjp}$ vanishes identically for $D=3$. For $D\geq 4$ we can try to achieve this condition by means of the gauge transformation  (\ref{10}).
\begin{prop}\label{p3}
 In dimension $D=3$ every solution and in $D\geq 4$ every analytic solution of  (\ref{3}) has the form
 \be\label{16c}
 T_{ij}=a\bigtriangleup R_{ij}-2a R^k_{(i,j)k}+bR_{,ij}+(aR^{kp}_{\ \ ,kp}-b\bigtriangleup R)g_{ij}\ ,
 \ee
 where $R_{ij}$ is an arbitrary symmetric tensor undergoing the gauge transformation
  \be\label{16d}
 R_{ij}\longrightarrow R_{ij}+\xi_{(i,j)}-\xi^k_{\ ,k} g_{ij}
 \ee
 with arbitrary functions  $\xi_i$.
\end{prop}
\begin{proof}
 Equation (\ref{16c}) follows  if (\ref{16a}) with $C^{ikjp}=0$ is substituted to (\ref{4}). In order to prove that for  $D\geq 4$ condition $C^{ikjp}=0$ is available one could look for an appropriate  gauge transformation of the form given by Proposition 2. In our opinion it is easier to  prove solvability of  (\ref{16c}) with respect to $R_{ij}$ if a transverse tensor $T^{ij}$ is given.
  Let us distinguish coordinate $x^1$, which together with  $x^a$,  $a=2,...D$, composes a Cartesian system of coordinates. Concerning evolution of $R_{ij}$  with respect to $x^1$ equations (\ref{16c}) with indices $11$ and $1a$ are constraints since they do not contain second derivatives of  $R_{ij}$ over $x^1$. It is easy to show that they are preserved in  $x^1$ if they are satisfied at  $x^1=0$ and remaining equations  (\ref{16c}) are fulfilled. If functions $T_{ij}$ are analytic then  the constraints at $x^1$ admit solutions and from the  Cauchy-Kowalewska theorem we obtain analytic solutions  $R_{ij}$ of all equations   (\ref{16c}). This situation is similar to that in general relativity. Equations (\ref{16c}) are identical with the linearized  Einstein equations if $R_{ij}-\frac{1}{D-1}Rg_{ij}$ is identified with the first corrections to the constant metric  $g_{ij}$. This analogy suggests 
  the gauge transformations (\ref{16d}).  It is easy to show that they preserve the rhs of (\ref{16c}). Counting number of components of  $R_{ij}$ and $\xi_i$ we can be sure that transformations  (\ref{16d}) are general up to functions of  $D-1$ variables. In order to exclude such additional functions one should find all transformations  (\ref{10}) preserving $C^{ikjp}=0$ and investigate their efect on $R_{ij}$. It hasn't been done.
 
\end{proof}
\noindent
\textbf{Remark:}
If  $g_{ij}$ has the  Lorentzian signature then equations (\ref{16c}) for $R_{ij}$ are hyperbolic and the assumption about analyticity in Proposition \ref{p3} is not necessary.

\section{TT tensors} 
For $T^{ij}$ given by expression (\ref{4}) the traceless condition (\ref{2})  leads to the following equation for the ``Ricci tensor''  $R^{ij}=R^{k}_{\ ikj}$
 \be\label{17}
 R^{kp}_{\ \ ,kp}=0\ .
 \ee
 Equation (\ref{17}) can be easily solved in terms of potentials.
  \begin{prop}\label{p4}
 For $D=2$ TT tensors are given by  (\ref{9b}),
 where $R$ satisfies
 \be\label{16b}
 \bigtriangleup R=0\ .
 \ee
  For $D\geq 3$ TT tensors are given by (\ref{4}) and
 \be\label{21}
 R^{pk}= S^{(pk)r}_{\ \ \ \ ,r}\ ,
 \ee
 where
 \be\label{19a}
 S^{pkr}=-S^{prk}\ .
 \ee
 \end{prop}
 \begin{proof} 
  For $D=2$ there is  $R^{ij}=\frac 12Rg^{ij}$, so (\ref{17}) reduces to equation (\ref{16b}), which can be further integrated by means of holomorphic functions (for a definite signature of 
 $g_{ij}$) or functions of null coordinates (for mixed signature).

 For $D\geq 3$ there is no other restrictions on  $R^{pk}$ except (\ref{17}). Integrating (\ref{17}) with use of  (\ref{0}) yields
 \be\label{18}
 R^{pk}_{\ \ ,k}=V^{pk}_{\ \ ,k}\ ,\ \ \ V^{pk}=-V^{kp}
 \ee
 and integrating (\ref{18}) leads to
  \be\label{19}
 R^{pk}=V^{pk}+S^{pkr}_{\ \ \ ,r}\ ,
 \ee
 where new potentials $ S^{pkr}$ satisfy (\ref{19a}).
 Symmetrization of (\ref{19}) over indices $pk$ yields (\ref{21}).
 
 \end{proof}
 Potentials $S^{ijk}$ are not uniquely defined. Their arbitrariness can be easily defined in the case of gauge  $C^{ijkp}=0$.
 \begin{prop}\label{p5}
  If (\ref{16c}) and (\ref{21}) are satisfied then  potentials $S^{ijk}$ undergo gauge transformations
 \be\label{19b}
 S^{ijk}\longrightarrow S^{ijk}-2g^{i[j}\xi^{k]}+\chi^{ijkr}_{\ \ \ \ ,r}+\eta^{ijk}\ ,
 \ee
 where
 \be\label{19e}
 \chi^{ijkr}=\chi^{i[jkr]}\ ,\ \ \eta^{ijk}=\eta^{[ijk]}\ .
 \ee
 \end{prop}
 \begin{proof}
 In order to prove  (\ref{19b}) one should solve condition
 \be\label{19c}
 S'^{(ij)k}_{\ \ \ \ \ ,k}=0
 \ee
 for an addition to   $S^{ijk}$ which does not influence $R^{ij}$ given by (\ref{21}).
 From (\ref{19c}) and (\ref{0}) it follows that
 \be\label{19d}
 S'^{(ij)k}=V^{ijkr}_{\ \ \ \ \ ,r}\ ,
 \ee
 where $V^{ijkr}=V^{(ij)[kr]}$. If we define $\chi^{ijkr}=\frac 38V^{i[jkr]}$ then $V^{ijkr}=4\chi^{(ij)kr}$. Taking into account the identity
 \be
 S'^{ijk}=S'^{[ijk]}+\frac 23S'^{(ij)k}-\frac 23S'^{(ik)j}
 \ee
 one obtains
 \be\label{20a}
 S'^{ijk}=\chi^{ijkr}_{\ \ \ \ ,r}-\chi^{[ijk]r}_{\ \ \ \ \ ,r}+S'^{[ijk]}\ .
 \ee
 Denoting two last expressions in (\ref{20a}) by $\eta^{ijk}$ leads to  transformation (\ref{19b}). 
 
 \end{proof}
 An application of  transformation (\ref{19b})  with $\xi^k=-S_i^{\ ik}/(D-1)$ and vanishing all other terms allows to obtain
  $S_i^{\ ik}=0$ and  $R=0$. Then substituting (\ref{21}) into  (\ref{16c}) yields 
  \be\label{19ba}
 T_{ij}=a(\bigtriangleup S_{(ij)\ ,k}^{\ \ \ k}+ S^{kr}_{\ \ (i,j)kr})\ .
 \ee
 An appropriate choice of  $\eta$ allows to get    $S^{[ijk]}=0$.
 Further gauge conditions can be assumed thanks to free functions  $\chi^{ijkr}$. 
 
 A description of TT tensors in dimension $D=3$ is much simpler than in $D\geq 4$. In this case Propositions \ref{p4} and \ref{p5} lead to the following result.
  \begin{prop}\label{p6}
  In dimension $D=3$ every TT tensor is given by
  \be\label{26}
 T^{ij}=\epsilon^{kl(i}(\bigtriangleup A_{\ k}^{j)}-A_{kp,}^{\ \ \ j)p})_{,l}\ ,
 \ee
 where $A_{ij}$ is a symmetric tensor undergoing the gauge transformations
  \be\label{26a}
 A_{ij}\longrightarrow A_{ij}+\chi_{(i,j)}+\eta g_{ij}
\ee
 with arbitrary functions $\chi_i$ and $\eta$.
  \end{prop}
\begin{proof}
 For $D=3$ one has $C^{ikjp}=0$ and
 \be\label{22}
 R^{ikjp}=-\epsilon^{ikl}\epsilon^{jps}G_{ls}\ ,
 \ee
 where $G_{ls}=R_{ls}-\frac 12 R g_{ls}$ corresponds to the Einstein tensor. Substituting (\ref{22}) to (\ref{4}) leads to the following simplified version of formula (\ref{16c}) 
 \be\label{22a}
 T^{ij}=-\epsilon^{ikl}\epsilon^{jps}G_{ls,kp}\ .
 \ee
 Gauge transformation (\ref{16d}) is now the most general since 
 \be\label{30}
 \xi^{ikjpr}=-\xi_l\epsilon^{lik}\epsilon^{jpr}
\ee 
 should be substituted into  (\ref{10}). If $T^i_{\ i}=0$ then $R^{pk}$ is given by  (\ref{21}), where
 \be\label{24}
 S^{pkr}=\epsilon^{krl}A^p_{\ l}\ .
 \ee
 Tensor $A$ can be arbitrary, but only its symmetric part appears in $G_{ls}$. Without loss of generality we can assume 
 \be\label{25}
 A^{lp}=A^{pl}\ ,
 \ee
 hence  the gauge condition $S_i^{\ ik}=0$ is satisfied.
 Substituting (\ref{24}) to (\ref{21}) and using (\ref{22a}) leads to (\ref{26}).
 
 Tensor $A$  can be further simplified by means of (\ref{19b}). Since for $D=3$
 \be\label{25a}
 \chi^{ijkr}=\chi^i\epsilon^{jkr}\ ,\ \ \eta^{ijk}=\eta\epsilon^{ijk}
 \ee
 condition $S_i^{\ ik}=0$ is preserved if
 \be\label{27}
 \xi^k=\frac 12\epsilon^{kij}\chi_{[i,j]}\ .
 \ee
Substituting (\ref{24}), (\ref{25a}) and (\ref{27}) to  (\ref{19b}) yields (\ref{26a}).

\end{proof}

\noindent
\textbf{Remark:}
If $T'^{ij}$ is a transverse tensor  then
 \be\label{31}
 T^{ij}=\epsilon^{kl(i}T'^{j)}_{\ \ k,l}
 \ee
 is a TT tensor. Case (\ref{26}) corresponds to
 \be\label{28}
 T'^{ij}=\bigtriangleup A^{ij}-2A^{(i\ \ j)p}_{\ \ p,}+A^{pk}_{\ \ ,pk}g^{ij}\ .
 \ee
 
 If $D\geq 4$ one can consider gauge conditions for a TT tensor other than $C^{ijkr}= 0$. A natural candidate is  condition $R^{ij}=0$.
 \begin{prop}\label{p6}
  In dimension $D\geq 4$ every analytic TT tensor has the form
  \be\label{32}
 T^{ij}=C^{ikjp}_{\ \ \ \ ,kp}\ ,
 \ee
 where $C^{ikjp}$ is a tensor  with all algebraic properties of the Weyl tensor.
  \end{prop}
\begin{proof}
 Formula (\ref{32}) is true if transformations (\ref{10}) allow to kill  $R^{ij}$. Let us look for solutions of equations
 \be\label{33}
\xi^{(kp)r} _{\ \ \ \ \ ,r}=-R^{kp}\ ,
 \ee
 where
 \be\label{33a}
 \xi^{kpr}=\xi^{ik\ pr}_{\ \ i}\ .
 \ee
 Let $x^1$ and $x^a$ be coordinates introduced in the proof of Proposition  \ref{p3}. Equation (\ref{33}) with $k=p=1$ is a constraint regarding evolution of
  $\xi^{kpr}$ with respect to $x^1$. The derivative of this constraint over $x^1$ leads to the new constraint
 \be\label{34}
 \xi^{(ab)1}_{\ \ \ \ \ ,ab}=R^{11}_{\ \ \ ,1}+2R^{1a}_{\ \ ,a}\ .
 \ee
 Condition (\ref{34}) is preserved if $T^i_{\ i}=0$ since then $R^{kp}_{\ \ ,kp}=0$. Thus, if we assume  (\ref{33}) with $k=p=1$ and (\ref{34}) on the surface  $x^1=0$ then, in analytic case,  we obtain solution of all equations  (\ref{33}). It remains to prove that for this solution there is a tensor  $\xi^{ikjpr}$ which satisfies (\ref{33a}). It is the case since every such tensor admits decomposition
 \be\label{35}
 \xi^{\ \ jpr}_{ik}=\hat\xi^{\ \ jpr}_{ik}+\alpha \delta^{[j}_{\ [i}\xi_{k]}^{\ pr]}-\beta \delta^{[j}_{\ [i}\delta_{k]}^{\ p}\xi^{r]}\ ,
 \ee
 where $\hat\xi^{\ \ jpr}_{ik}$ is traceless and
 \be
 \alpha=\frac{6}{D-3}\ ,\ \ \beta=\frac{6}{(D-2)(D-3)}\ ,\ \ \xi^r=\xi_k^{\ kr}\ .
 \ee
 It follows from  (\ref{35}) that tensors $\hat\xi^{ikjpr}$ and $\xi^{kpr}$ can be defined independently.

\end{proof}

 If $R^{ij}=0$ then gauge transformations  (\ref{10}) are restricted by\\ $\xi^{(kp)r} _{\ \ \ \ \ ,r}=0$. It follows from (\ref{19c})-(\ref{20a}) that solution of this condition is
 \be\label{36}
 \xi^{ijk}=\chi^{ijkr}_{\ \ \ \ ,r}+\eta^{ijk}\ ,
 \ee
 where tensors $\chi$ and $\eta$ satisfy (\ref{19e}). 
 Transformations  (\ref{10}) with condition (\ref{36}) can be further used to simplify formula (\ref{32}).
 
 \section{Discussion}
 All calculations were performed in  Cartesian coordinates $x^i$. If we replace  partial derivatives with respect to $x^i$ by covariant derivatives  $\nabla_k$ then we can pass to other systems of coordinates. In particular formula (\ref{4}) then takes the form
 \be\label{37}
 T^{ij}=\nabla_k\nabla_pR^{ikjp}\ .
 \ee
 For a nonflat metric expression (\ref{37})   fails to satisfy (\ref{1}) because of nonvanishing commutators of covariant derivatives. 
 
 Given solutions of equations (\ref{1}) and (\ref{2}) in a flat space we can obtain all TT tensors in conformally flat spaces by means of transformation (\ref{1a}).
 The author is not able to express in general these solutions in terms of covariant derivatives of potentials  with no use of the conformal factor $\psi$. The only exception is the case of spaces of constant curvature (space forms corresponding to flat space, sphere or hyperbolic space). Then
  \be\label{39}
 R^{ij}_{\ \ kp}=\lambda (\delta^i_{\ k}\delta^j_{\ p}-\delta^j_{\ k}\delta^i_{\ p})\ ,\ \ \lambda=const
 \ee
and equation (\ref{1}) is  still  satisfied if formula (\ref{37}) is replaced by
 \be\label{40}
 T^{ij}=-\nabla_k\nabla_p R^{k(ij)p}+\frac 12\lambda R^{ij}\ .
 \ee
 In order to satisfy $T^i_{\ i}=0$ it is sufficient, but may be not necessary, to assume the covariant version of (\ref{21})
 \be\label{21a}
 R^{pk}= \nabla_rS^{(pk)r}
 \ee
 together with 
 \be
 S^{ijk}=-S^{ikj}\ ,\ \ S_i^{\ ik}=0\ .
 \ee
 We expect that for $D\geq 4$ the gauge condition $C^{ijkl}=0$ is available for transverse tensors and condition $R^{ij}=0$ for TT tensors.

 Let us assume that space is flat but a TT tensor is invariant under a symmetry of metric. Such solutions can be obtained by means of invariant tensor potentials appearing in Propositions 1-7. Such description is not necessarily optimal as it is shown in the paper of 
  Conboye and \'{O} Murchadha \cite{cm} and Conboye \cite{c} in dimension $D=3$. Their expressions for  TT tensors contain only two arbitrary functions and at most  their second derivatives, not third as in 
  (\ref{26}).  It is difficult to say if a similar simplified description exists in higher dimensions.

\end{document}